\documentclass[11pt]{article}

\usepackage[a4paper,top=2.2cm,bottom=2.2cm,left=2.2cm,right=2.2cm,marginparwidth=1.75cm]{geometry}
\usepackage{color}
\usepackage{hyperref}
\hypersetup{
     colorlinks = true,
     linkcolor = blue,
     anchorcolor = blue,
     citecolor = blue,
     filecolor = blue,
     urlcolor = blue
     }
     
\usepackage{tikz}
\usepackage{pgfplots}
\usetikzlibrary{fillbetween}
\pgfplotsset{compat=newest}

\usepackage[cmex10]{amsmath}
\usepackage{amssymb}
\usepackage{amsmath}
\usepackage{amsthm}
\DeclareMathOperator*{\argmin}{argmin}

\usepackage{epsfig}
\usepackage{graphicx}
\usepackage{graphics}
\usepackage{subfigure}
\usepackage{xspace}

\usepackage{float}

\usepackage{cite}
\usepackage{eurosym}
\usepackage{enumitem}
\makeatletter\DeclareRobustCommand\onedot{\futurelet\@let@token\@onedot}
\def\@onedot{\ifx\@let@token.\else.\null\fi\xspace}

\newtheorem{proposition}{Proposition}

\begin{document}

\title{\huge {\bf Trading Data for Wind Power Forecasting: A Regression Market with Lasso Regularization}}

\author{\Large Liyang Han$^{\dag *}$, Pierre Pinson$^\ddag$, Jalal Kazempour$^\dag$ \\ \textcolor{white}{.} \\ \small $^\dag$ Technical University of Denmark \\ \small Department of Wind and Energy Systems \\ \small $^*$ Corresponding Author Email: \href{mailto:liyha@dtu.dk}{liyha@dtu.dk} \\
\textcolor{white}{.} \\ \small $^\ddag$ Technical University of Denmark \\ \small Department of Technology, Management and Economics  \vspace{1cm}}

\date{\today}

\maketitle

\begin{abstract}
This paper proposes a regression market for wind agents to monetize data traded among themselves for wind power forecasting.  Existing literature on data markets often treats data disclosure as a binary choice or modulates the data quality based on the mismatch between the offer and bid prices. As a result, the market disadvantages either the data sellers due to the overestimation of their willingness to disclose data,  or the data buyers due to the lack of useful data being provided. Our proposed regression market determines the data payment based on the least absolute shrinkage and selection operator (lasso), which not only provides the data buyer with a means for selecting useful features, but also enables each data seller to individualize the threshold for data payment. Using both synthetic data and real-world wind data, the case studies demonstrate a reduction in the overall losses for wind agents who buy data, as well as additional financial benefits to those who sell data.

\vspace{1cm}


\end{abstract}

\newpage

\section{Introduction}
\label{sec:intro}

As the complexity and uncertainty of modern energy systems grow, the value of data for improving system and market operations has become a focus of academic research in recent years \cite{Veldkamp2020}.  One important usage of data is to inform the forecast of loads and generation \cite{Morales2014}.  Focusing on wind power forecasting, this paper investigates how the value of external data can be quantified and monetized.

The added value of data for wind power forecasting has been well examined under the assumption that data from external sources is a free and highly accessible commodity \cite{Jung2014, Tastu2010, Andrade2017}. However, this assumption becomes unrealistic in applications where data privacy is highly valued \cite{Farokhi2020} or where a conflict of interest exists between the data owners and the data users \cite{Goncalves2020}.  

To incentivize data exchange, recent literature has been exploring the idea of a data market, in which the data owners (sellers) receive financial compensation from the data users (buyers) for their disclosure of data \cite{Bergemann2019}. In general, these market frameworks can be put into two categories: buyer-centric markets and seller-centric markets. In a buyer-centric market,  the data buyer has full control over the data price, while each data seller is either assumed to always accept the offer \cite{Pinson2022}, or given a binary choice on the offer while suffering a privacy loss \cite{Montes2019, Acemoglu2021}. By contrast, a seller-centric market gives the data sellers the authority to directly add noise to the data \cite{Bimpikis2019}, or has a third party set the price  based on the added value of previously traded data \cite{Agarwal2019}. In the latter, noise is also added to the data if the offering price from the data buyer is lower than the set price. As these market frameworks are often based on game theory, it is often computationally intensive to find the market equilibrium in a noncooperative game setting \cite{Montes2019, Acemoglu2021, Bimpikis2019} or to derive the market payoffs in a cooperative game setting \cite{Agarwal2019}.

In the context of wind power forecasting, it has been shown in the literature that data from surrounding sites, such as other wind farms, can help improve forecast accuracy substantially due to the significant spatio-temporal correlations \cite{Tastu2010,  Girard2013}. A wind agent's losses from forecast errors are commonly measured by the mean squared error (MSE) of the forecast compared to the target values of power generation \cite{He2014, Pinson2016}. Using the reduction of MSE as the measurement for a wind agent's profit as a data buyer, \cite{Goncalves2020} and \cite{Pinson2022} adopt the ordinary least squares (OLS) regression to estimate the forecasting parameters, and construct data markets to incentivize wind agents as market competitors to trade data among themselves. The framework in \cite{Goncalves2020} adopts the pricing scheme from \cite{Agarwal2019}. where the price of data in each trade is determined by the added value of the sellers' data to the previous buyers. As a result, the market outcomes are dependent on the ordering of the buyers, leading to potential large suboptimality gaps and unfairness. Instead of pricing the data prior to each trade, the framework in \cite{Pinson2022} quantifies the reduction of the losses for each data buyer and allocates the resulting savings to the sellers based on their contribution to the buyer's task. The caveat of this market is the assumption that the sellers have no revenue requirements, when in reality this might not hold true.

OLS regression is known to be highly sensitive to outliers and prone to overfitting when the supporting features, sets of data that are fitted to the target variable in the regression, are highly correlated \cite{Boyd2010}. Since wind power data for closely located wind farms are likely highly correlated, feature selection methods such as the least absolute shrinkage and selection operator (lasso) \cite{Tibshirani1996} have been widely adopted for parameter estimation in wind power forecasting \cite{Pinson2016,  Cavalcante2017,  Messner2019}. It helps eliminate the impact of less correlated data in the forecasting task through adding to the loss function the product of a positive regularization parameter and the absolute value of the coefficient for each feature, referred to as the lasso term. 

We observe that not only is the lasso term effective in preventing overfitting, but it can also provide a tool for thresholding the selection of features through customizing their corresponding regularization parameters: the higher the regularization parameter, the more correlated a feature has to be to the target to be selected. This observation serves as the inspiration for our proposed market framework.

The main contribution of this paper is the design of a \emph{lasso regression market} for wind power forecasting, in which the lasso term is exploited as a measure for direct payments from data buyers to data sellers. This market framework has three main advantages: 1) compared to the seller-centric model, the lasso provides the buyers with better quality data through selecting from the sellers' data the features that are most likely to improve the forecast without additional noise; 2) compared to the buyer-centric model, sellers are given the authority to threshold their financial reward through individualizing the lasso regularization parameter; and 3) compared to data markets that are set up based on game theory, the computation of the market outcomes is simpler as the lasso term as an output from the regression directly represents the payment. 

The paper is structured as follows: Section~\ref{sec:setup} introduces the market participants. Section~\ref{sec:regframe} describes the OLS regression and the lasso regression for the data buyer's analytics task. Section~\ref{sec:lasso} formulates the lasso regression market for wind agents and proves its financial viability for both data buyers and data sellers. Section~\ref{sec:cases} uses synthetic data and measured data from Nord Pool\footnote{https://www.nordpoolgroup.com/Market-data1/Power-system-data/Production1/Wind-Power/ALL/Hourly1/} to demonstrate the effectiveness of the proposed market. Section~\ref{sec:conclusion} concludes the paper with discussions on the key findings and future work.

\section{Market Participants}
\label{sec:setup}

The use case for the proposed regression market is data trading among wind agents. The data buyer is called a \emph{central agent} who has an analytics task to estimate parameters for forecasting wind power. The other wind agents are data sellers, referred to as \emph{support agents},  hold data that could potentially improve the central agent's forecast, and expect to be remunerated for sharing those data.  We focus on a single-buyer setup, but it can be easily extended to multiple buyers since the agents' analytics tasks are independent of each other in the sense that the outcome of a central agent's forecasting task does not affect any concurrent task of another agent.

\subsection{Agents and the Analytics Task}

We denote the full set of wind agents by $\mathcal{N} = \{1, 2, \ldots, N\}$, indexed by $i$, so $i \in \mathcal{N}$. The data buyer, or the central agent $c \in \mathcal{N}$, has an analytics task to estimate parameters for forecasting wind power as a time series target variable, denoted by vector $\mathbf{y}_c = \left[y_{c, 1} \ y_{c, 2} \ \ldots \ y_{c,T} \right]^\top$, where $T$ is the total number of time steps of concern. The support agents are gathered in set $\mathcal{N}_{- c} = \mathcal{N} \setminus \{c\} = \{i \in \mathcal{N} | i \neq c\}$, and the data owned by each agent $i$ are in set $\mathcal{X}_i = \{\mathbf{x}_i^d | d \in [1, D_{i \rightarrow c}]\}$, where $D_{i \rightarrow c}$ is the total number of relevant features from agent $i$ for the central agent's analytics task, and vector $\mathbf{x}_i^d$ represents agent $i$'s $d$th feature (e.g., $l$-hour-ahead wind power data: $\mathbf{x}_i^d = [x_{i,1-l}^d \ x_{i,2-l}^d \ \ldots \ x_{i,T-l}^d]^\top$). The analytics task with all agents' data is given by
\begin{equation}
   \hat{\mathcal{B}}_\mathcal{N} = F \left(\mathcal{X}_\mathcal{N},\mathbf{y}_c \right) ,
   \label{eq:beta_all}
\end{equation}
where $\mathcal{X}_\mathcal{N} = \cup_{i \in \mathcal{N}} \mathcal{X}_i$ gathers the relevant input features of the central agent and support agents, and $\hat{\mathcal{B}}_\mathcal{N}$ is the set of estimated parameter values for the analytics task. If the analytics task follows a linear regression framework, $\hat{\mathcal{B}}_\mathcal{N} = \{\beta_i^d | i \in \mathcal{N}, d = 1, \ldots,D_{i \rightarrow c}\} \cup \{\beta_\mathcal{N}^0\}$ is the set of linear coefficients, with each element corresponding to a feature of an agent and $\beta_\mathcal{N}^0$ representing the intercept term. 

Similarly, if the central agent only relies on their own data for the analytics task, it solves
\begin{equation}
   \hat{\mathcal{B}}_c = F \left(\mathcal{X}_c,\mathbf{y}_c \right) ,
   \label{eq:beta_self}
\end{equation}
where $\hat{\mathcal{B}}_c = \{\beta_c^d | d = 0, 1, \ldots,D_{c \rightarrow c}\}$ gathers the linear coefficients corresponding to the features owned by the central agent, with $\beta_c^0$ being the intercept. 
Using results from \eqref{eq:beta_all} and \eqref{eq:beta_self}, the central agent can forecast the target variable as $\hat{\mathbf{y}}_c \left(\mathcal{X}_\mathcal{N}, \hat{\mathcal{B}}_\mathcal{N} \right)$ and $\hat{\mathbf{y}}_c \left(\mathcal{X}_c, \hat{\mathcal{B}}_c \right)$, respectively. Details on linear regression for the analytics task and forecasting are explained in Section \ref{sec:regframe}.

\subsection{Support Agents and Reservation to Sell Data} \label{subs:support_unwill}

Note that the support agents and the central agent can be competitors in the same energy market, which may discourage the support agents from disclosing their data. We measure this barrier in the form of a payment threshold denoted by 
\begin{equation}
H_i^d \left(u_i^d,  \beta_i^d \right) = \left| u_i^d \beta_i^d \right| ,
\end{equation}
representing the payment required by support agent $i \in \mathcal{N}_{-c}$ for disclosing data associated with their $d$th feature. It is a function of $u_i^d$ and $\beta_i^d$, measuring, respectively, the reservation of agent $i$ to sell their $d$th feature, and how correlated this feature is to the central agent's target variable. Therefore, the higher $u_i^d$ and $\beta_i^d$ are, the higher the payment needs to be for agent $i$ to disclose their $d$th feature to the central agent. In other words, the payment threshold increases as the support agent becomes less willing to sell their data, and as the correlation between the support agent's data and the central agent's target variable becomes stronger.

Some factors that a support agent may take into account when determining their revenue threshold includes the valuation of their loss of privacy, the valuation of their losses due to an increase in their competitor's profit, the cost of collecting data and offering into the regression market, etc. In our proposed framework, the revenue threshold is non-negative. In practice, however, there could be scenarios where the support agent is able to benefit directly from the improvement of the central agent's analytics task (e.g., an overall improvement of the forecast of renewable generation in the energy market can lead to a decrease in the imbalance price that is eventually passed down to all the renewable agents), and is thus willing to receive a negative payment in the regression market. We choose to leave the latter scenario for future work.

\subsection{Central Agent and Willingness to Buy Data} \label{subs:cen_will}
We assume that the central agent uses the MSE to measure the average losses from the forecast for each time step:
\begin{equation} \label{eq:central_will}
S_c^\text{MSE} \left(\hat{\mathbf{y}}_c\right) =  \frac{1}{T} \sum_{t=1}^T \left(y_{c,t} - \hat{y}_{c,t} \right)^2 ,
\end{equation}
where $\hat{\mathbf{y}}_c = \left[\hat{y}_{c,1} \ \hat{y}_{c,2} \ \ldots \ \hat{y}_{c,T} \right]^\top$.  In practice, the central agent may assign a scaling factor to the MSE to represent their actual losses in monetary terms, but we assume this scaling factor to be ``$1$'' in this paper for simplicity. Therefore, $S_c^\text{MSE}$ can be considered to directly represent the central agent's mean financial losses at each time step.

Given the payment threshold $H_i^d$ to obtain agent $i$'s $d$th feature, in order for the central agent to financially benefit from purchasing data in the regression market, it requires
\begin{equation} 
S_c^\text{MSE} \left( \hat{\mathbf{y}}_c \left(\mathcal{X}_c, \hat{\mathcal{B}}_c \right) \right) - S_c^\text{MSE} \left( \hat{\mathbf{y}}_c \left(\mathcal{X}_\mathcal{N}, \hat{\mathcal{B}}_\mathcal{N} \right) \right) \geq \sum_{i\in \mathcal{N}_{- c}} \sum_{d=1}^{D_{i \rightarrow c}} H_i^d .
\label{eq:central_IR}
\end{equation}

\section{Analytics Task under Linear Regression}
\label{sec:regframe}

Given \eqref{eq:central_IR}, we anticipate there to be an opportunity for the central agent to purchase data from the support agents while making sure the total payment is lower than the reduced losses for the forecast. In this section, we define the analytics tasks in different forms of linear regression for both the case with only the central agent's own data and the case with additional data from the support agents.

\subsection{Forecast with Linear Coefficients}
\label{ssec:linreg}

In the case where the central agent $c$ only considers their own features, the forecast of their target based on the coefficients from a linear regression can be written as
\begin{equation} \label{eq:lincentral}
    \hat{y}_{c,t} = \beta_c^0 + \sum_{d=1}^{D_{c \rightarrow c}} \beta_c^d x_{c,t}^d  \ ,
\end{equation}
In contrast, if the features of the support agents were also considered, the forecast would become
\begin{align} 
    \hat{y}_{c,t} & = \underbrace{\beta_c^0 + \sum_{d=1}^{D_{c \rightarrow c}} \beta_c^d x_{c,t}^d}_{\text{central agent}} \ + \underbrace{\sum_{i \in \mathcal{N}_{- c}} \left( \beta_i^0 + \sum_{d=1}^{D_{i \rightarrow c}} \beta_i^d x_{i, t}^d \right)}_{\text{support agents}} \nonumber \\
    & = \beta_\mathcal{N}^0 + \sum_{i \in \mathcal{N}} \sum_{d=1}^{D_{i \rightarrow c}} \beta_i^d x_{i, t}^d  \ , \label{eq:linall}
\end{align}
where $\beta_\mathcal{N}^0 = \sum_{i \in \mathcal{N}} \beta_{i,0}$ is the sum of the intercept terms of all agents, and thus the overall intercept. 

In the rest of the paper we assume that the dependencies of the target variable on the features are stationary, and thus the true linear coefficients do not vary with time. For ease of notation, we define vector $\mathbf{x}_{i,t} = [x_{i,t}^1 \ x_{i,t}^2 \ \ldots \ x_{i,t}^{D_{i \rightarrow c}}]^\top, \forall i \in \mathcal{N}$ as the values of all of agent $i$'s features used for the forecast of the central agent's target variable at time step $t$.

\subsection{Baseline Losses using OLS Regression on Own Features}
\label{ssec:ols_reg}

Without a data market, the central agent can only rely on their own data for the analytics task and forecasting. We construct the central agent's own regressor vector $\mathbf{x}_{\{c\}, t} = [1 \ \mathbf{x}_{c,t}^\top]^\top$, and use $\boldsymbol{\beta}_c^\ast = [{\beta_c^0}^\ast \ {\beta_c^1}^\ast \ \ldots \ {\beta_c^{D_{c \rightarrow c}}}^\ast]^\top$ to denote the vector of linear coefficients for the regressors. Using OLS regression, the analytics task in \eqref{eq:beta_self} becomes
\begin{align} 
    \boldsymbol{\beta}_c^\ast & =  \argmin_{\boldsymbol{\beta} \in \mathbb{R}^{1+D_{c \rightarrow c}}} \ \sum_{t=1}^T \left(y_{c,t} - \hat{y}_{c,t}\right)^2 \label{eq:ols_self} \\
    & \stackrel{\eqref{eq:lincentral}}{=}  \argmin_{\boldsymbol{\beta} \in \mathbb{R}^{1+D_{c \rightarrow c}}} \ \sum_{t=1}^T \left(y_{c,t} - \mathbf{x}_{\{c\},t}^\top \boldsymbol{\beta}\right)^2  .\label{eq:lscentral}
\end{align}

Comparing \eqref{eq:ols_self} and \eqref{eq:central_will}, since the constant $\frac{1}{T}$ does not affect the ``$\argmin$'' function, we can rewrite \eqref{eq:ols_self} as
\begin{equation} \label{eq:ols_mse}
     \boldsymbol{\beta}_c^\ast = \argmin_{\boldsymbol{\beta} \in \mathbb{R}^{1+D_{c \rightarrow c}}} S_c^\text{MSE} \left(\hat{\mathbf{y}}_c \left(\mathcal{X}_c, \boldsymbol{\beta}\right) \right) .
\end{equation}

Since the MSE measures the central agent's financial losses from the forecast (Section \ref{subs:cen_will}), we can consider $S_c^\text{MSE} \left(\hat{\mathbf{y}}_c \left(\mathcal{X}_c, \boldsymbol{\beta}_c^\ast \right) \right)$ the baseline losses of the central agent without access to data from other agents.

\subsection{Losses using Lasso Regression with Support Features} 
\label{ssec:lasso_reg}

When data from the support agents are made available to the central agent, there is a potential for the central agent to improve their forecasting accuracy. An OLS regression can again be applied to fit the available features, but a well known drawback of OLS regression is its sensitivity to outliers and highly correlated features,  which tends to result in overfitting the data and compromising the forecast accuracy \cite{Boyd2010}. Additionally, if the data collected from the support agents are of low quality (e.g., containing missing values, voluntarily flawed, etc.), OLS regression is not capable of eliminating these corrupted features. One way to mitigate such risks is to apply cross-validation of the support features, but it can be computationally expensive and complicates the market set up (i.e., how to reward the support agents for the portion of their data that are only used for cross-validation). Another way is to use feature selection methods, including $L_p$ regularization methods that have been proposed to reduce the impact of less correlated features on the forecast performance. Among these, the lasso is a popular regularization method that yields sparse coefficient matrices, which helps prevent overfitting \cite{Tibshirani1996}. The lasso term is an $L_1$-norm penalty applied to the coefficients of a linear regression problem.  

To implement the lasso regression on the features of all agents, we construct the all-agents regressor vector $\mathbf{x}_{\mathcal{N},t} = [1 \ \mathbf{x}_{1,t}^\top \ \ldots \mathbf{x}_{N,t}^\top]^\top$. We then denote the vector of the corresponding linear coefficients using lasso regression by $\boldsymbol{\beta}_\mathcal{N}^{L_1} = [{\beta_\mathcal{N}^0}^{L_1} \ {\boldsymbol{\beta}_{\mathcal{N},1}^{L_1}}^\top \ \ldots \ {\boldsymbol{\beta}_{\mathcal{N},N}^{L_1}}^\top]^\top$, where $\boldsymbol{\beta}_{\mathcal{N},i}^{L_1} = [{\beta_i^1}^{L_1} \ {\beta_i^2}^{L_1} \ \ldots \ {\beta_i^{D_{i \rightarrow c}}}^{L_1}]^\top, \forall i \in \mathcal{N}$. Applying the generic lasso estimator, the analytics task in \eqref{eq:beta_all} becomes
\begin{equation} \label{eq:lasso_all}
    \boldsymbol{\beta}_\mathcal{N}^{L_1} = \argmin_{\boldsymbol{\beta} \in \mathbb{R}^{1+\sum_{i \in \mathcal{N}} D_{i \rightarrow c}}} \  \left[ \frac{1}{2} \sum_{t=1}^T \left(y_{c,t} - \mathbf{x}_{\mathcal{N},t}^\top \boldsymbol{\beta} \right)^2 + \lambda \left\| \boldsymbol{\beta} \right\|_1 \right] ,
\end{equation}
where $\lambda$ is the lasso regularization parameter. The lasso term $\lambda \left\| \boldsymbol{\beta} \right\|_1$ shrinks some coefficients and sets some of them to zero. As a result, it reduces or even eliminates the influence of the features that are less likely to contribute to the improvement of the forecast. The greater $\lambda$ is, the more shrinkage is applied to the regression coefficients.  

We apply a constant factor of $\frac{2}{T}$ within the ``$\argmin$'' function in \eqref{eq:lasso_all} and rewrite it based on \eqref{eq:linall} and \eqref{eq:central_will} as
\begin{equation} \label{eq:lasso_all_mse}
    \boldsymbol{\beta}_\mathcal{N}^{L_1} = \argmin_{\boldsymbol{\beta} \in \mathbb{R}^{1+\sum_{i \in \mathcal{N}} D_{i \rightarrow c}}} \  \left[ S_c^\text{MSE} \left(\hat{\mathbf{y}}_c \left(\mathcal{X}_\mathcal{N}, \boldsymbol{\beta} \right) \right) + \frac{2 \lambda}{T} \left\| \boldsymbol{\beta} \right\|_1 \right] .
\end{equation}
Similarly, $S_c^\text{MSE} \left(\hat{\mathbf{y}}_c \left(\mathcal{X}_\mathcal{N}, \boldsymbol{\beta}_\mathcal{N}^{L_1} \right) \right)$ represents the financial losses of the central agent if the forecast is based on a lasso regression given all support agents' features.

\section{Regression Market with the Lasso} 
\label{sec:lasso}

In this section, we introduce the concept of the \emph{lag} in time series data as a way to define the number of relevant features from recent time steps from each agent. For wind power forecasting, the lag not only captures the temporal correlations of the wind generation at a specific site, it also indirectly captures the spatial correlations between neighboring sites as a result of the natural development of wind \cite{Girard2013}. Using the lasso term to define the payment, we then construct a regression market for wind agents, which is proved to meet the payment requirement of support agents while guaranteeing profit for the central agent. We note that only in-sample market outcomes are analyzed in this paper, meaning $T$ can also be interpreted as the total number of time steps for training the model.

\subsection{Linear Regression on Features with a Fixed Maximum Lag}
\label{ssec:weightedlasso}
Recall from Sections \ref{sec:setup} and \ref{sec:regframe} that $D_{i \rightarrow c}$ represents the number of relevant features from each agent. Here, since we are dealing with time series data, we equate $D_{i \rightarrow c}$ to agent $i$'s maximum lag, which is the number of recent time steps that are considered relevant for the target variable.  This means that each feature from a support agent represents their data of a certain lag to the target variable. Assuming features from all agents are available up to one time step before the target variable $y_{c,t}$ is revealed, and that a fixed maximum lag $\mathcal{D}$ applies to all agents, i.e., $D_{i \rightarrow c} = \mathcal{D}, \forall i \in \mathcal{N}$, the relevant features for $y_{c,t}$ from any agent $i$ can be gathered in the vector $\mathbf{x}_{i,t} \in \mathbb{R}^\mathcal{D}= [x_{i,t-\mathcal{D}} \ x_{i,t-\mathcal{D}+1} \ \ldots \ x_{i,t-1}]^\top, \forall i \in \mathcal{N}$. Applying the fixed maximum lag $\mathcal{D}$ to the regressions in Sections \ref{ssec:ols_reg} and \ref{ssec:lasso_reg},  we have $|\boldsymbol{\beta}_c^\ast| = 1 + \mathcal{D}$, and $|\boldsymbol{\beta}_\mathcal{N}^\ast| = 1 + N \mathcal{D}$.

The features of any agent $i$ can then be expressed in $\mathbf{X}_{i,T} \in \mathbb{R}^{T \times \mathcal{D}} := \left[\mathbf{x}_{i,1} \ \mathbf{x}_{i,2} \ \ldots \ \mathbf{x}_{i,T}\right]^\top$. To prepare for the regression, we gather the data from the central agent alone in $\mathbf{X}_{\{c\},T} \in \mathbb{R}^{T \times (1 + \mathcal{D})} := \left[\mathbf{1}_T \ \mathbf{X}_{c,T}\right]$ and data from all agents in $\mathbf{X}_{\mathcal{N},T} \in \mathbb{R}^{T \times (1 + N \mathcal{D})} := \left[\mathbf{1}_T \ \mathbf{X}_{1,T} \  \ldots \ \mathbf{X}_{N,T}\right]$, where $\mathbf{1}_T = [\underbrace{1 \ \ldots \ 1}_T]^\top$. Therefore, \eqref{eq:lscentral} and \eqref{eq:lasso_all} can be rewritten as 
\begin{equation}
\boldsymbol{\beta}_c^\ast  = \argmin_{\boldsymbol{\beta} \in \mathbb{R}^{1 + \mathcal{D}}} \left\| \mathbf{y}_c - \mathbf{X}_{\{c\},T} \boldsymbol{\beta} \right\|_2^2 ,
\label{eq:multi_t_cen_ols}
\end{equation}
and
\begin{equation}
\boldsymbol{\beta}_\mathcal{N}^{L_1} = \argmin_{\boldsymbol{\beta} \in \mathbb{R}^{1 + N \mathcal{D}}} {\left\{ \frac{1}{2} \left\| \mathbf{y}_c -  \mathbf{X}_{\mathcal{N},T} \boldsymbol{\beta} \right\|_2^2 + \lambda \left\| \boldsymbol{\beta} \right\|_1 \right\}} \ .
\label{eq:multi_t_all_lasso}
\end{equation}

For clarity, we expand the following matrix operation in \eqref{eq:multi_t_all_lasso} as
\begin{equation*}
\mathbf{y}_c - \mathbf{X}_{\mathcal{N},T} \boldsymbol{\beta} = 
\begin{bmatrix}
y_{c,1} \\
y_{c,2} \\
 \vdots \\
y_{c,T}
\end{bmatrix}
-
\begin{bmatrix}
1 & 1 & \cdots & 1 \\
x_{1,1-\mathcal{D}} &  x_{1,2-\mathcal{D}}  & \cdots & x_{1,T-\mathcal{D}} \\
x_{1,1-\mathcal{D}+1} &  x_{1,2-\mathcal{D}+1} & \cdots &  x_{1,T-\mathcal{D}+1} \\
\vdots & \vdots & \ddots & \vdots \\
x_{1,0} &  x_{1,1}& \cdots & x_{1,T-2} \\
\vdots & \vdots & \ddots & \vdots \\
x_{N,1-\mathcal{D}} &  x_{N,2-\mathcal{D}}  & \cdots & x_{N,T-\mathcal{D}} \\
x_{N,1-\mathcal{D}+1} &  x_{N,2-\mathcal{D}+1} & \cdots &  x_{N,T-\mathcal{D}+1} \\
\vdots & \vdots & \ddots & \vdots \\
x_{N,0} &  x_{N,1}& \cdots & x_{N,T-2} \\
\end{bmatrix}^\top
\begin{bmatrix}
\beta_\mathcal{N}^0 \\
\beta_1^1 \\
\beta_1^2 \\
\vdots \\
\beta_1^\mathcal{D} \\
\vdots \\
\beta_N^1 \\
\beta_N^2 \\
\vdots \\
\beta_N^\mathcal{D} \\
\end{bmatrix} .
\end{equation*}

\subsection{Regression Market with the Lasso Term as Payment}

Considering all $T$ time steps, we rewrite the average financial losses of the central agent as a function of the regressor matrix $\mathbf{X}$ and a given vector of linear coefficients $\boldsymbol{\beta}$ as
\begin{equation}
S_c^\text{MSE} \left(\mathbf{X} \in \mathbb{R}^{T \times |\boldsymbol{\beta}|}, \boldsymbol{\beta} \right) =  \frac{1}{T} \left\| \mathbf{y}_c - \mathbf{X} \boldsymbol{\beta} \right\|_2^2 \ ,
\label{eq:mse_mat}
\end{equation}
where $|\boldsymbol{\beta}|$ equals the number of features available for the regression including the intercept.

In the generic lasso estimator in \eqref{eq:multi_t_all_lasso}, a single $\lambda$ is applied to all the coefficients, whereas in practice we can assign different $\lambda$'s to different coefficients without compromising the computation efficiency.  Therefore we can construct a lasso regularization scalar matrix $ \boldsymbol{\lambda}_\mathcal{N}^\mathcal{D} \in \mathbb{R}^{1 + N \mathcal{D}}_{\geq 0} =  \text{diag}(\lambda_0 \ \lambda_1^1 \ \ldots \ \lambda_1^\mathcal{D} \ \lambda_2^1 \ \ldots \ \lambda_2^\mathcal{D} \ \ldots \ \lambda_N^1 \ \ldots \ \lambda_N^\mathcal{D})$,  where we set $\lambda_0 = 0$ to ensure no shrinkage is applied to the intercept term.  We then define the \emph{lasso loss} function as
\begin{align}
S_c^{L_1} \left(\mathbf{X}_{\mathcal{N},T}, \boldsymbol{\lambda}_\mathcal{N}^\mathcal{D}, \boldsymbol{\beta} \right) & = \frac{1}{T} \left\| \mathbf{y}_c -  \mathbf{X}_{\mathcal{N},T} \boldsymbol{\beta} \right\|_2^2 + \left\| \frac{2}{T} \boldsymbol{\lambda}_\mathcal{N}^\mathcal{D} \boldsymbol{\beta} \right\|_1 \nonumber \\
& = S_c^\text{MSE} \left(\mathbf{X}_{\mathcal{N},T}, \boldsymbol{\beta} \right) +  \left\| \frac{2}{T} \boldsymbol{\lambda}_\mathcal{N}^\mathcal{D} \boldsymbol{\beta} \right\|_1 .
\label{eq:lasso_loss} 
\end{align}
So \eqref{eq:multi_t_all_lasso} can be modified as 
\begin{equation}
\boldsymbol{\beta}_\mathcal{N}^{L_1} \left(\mathbf{X}_{\mathcal{N},T}, \boldsymbol{\lambda}_\mathcal{N}^\mathcal{D} \right) = \argmin_{\boldsymbol{\beta} \in \mathbb{R}^{1 + N \mathcal{D}}} S_c^{L_1} \left( \mathbf{X}_{\mathcal{N},T}, \boldsymbol{\lambda}_\mathcal{N}^\mathcal{D}, \boldsymbol{\beta} \right) .
\label{eq:reg_market_lasso}
\end{equation}

Here, we propose to use \eqref{eq:reg_market_lasso} as the basis for the regression market for wind agents.  Relying on their own data, the central agent has the baseline financial losses of $S_c^\text{MSE} \left(\mathbf{X}_{\{c\},T},  \boldsymbol{\beta}_c^\ast \right)$. When support agent $i \in \mathcal{N}_{-c}$ offers feature $x_{i, t}^d$ into the market, they also have the freedom to set $\lambda_i^d$ based on $u_i^d$, their reservation to sell  (Section \ref{subs:support_unwill}).  After the market operator conducts the market using \eqref{eq:reg_market_lasso}, ${\beta_i^d}^{L_1}$ is computed. If ${\beta_i^d}^{L_1} = 0$, feature $x_{i, t}^d$ is not selected to be used for the central agent's forecast, and no payment is needed. Otherwise the central agent has to pay agent $i$ in the amount of $\left| \frac{2}{T} \lambda_i^d {\beta_i^d}^{L_1} \right|$ for using feature $x_{i, t}^d$ for the forecast. Therefore, the lasso term $\left\| \frac{2}{T} \boldsymbol{\lambda}_\mathcal{N}^\mathcal{D} \boldsymbol{\beta} \right\|_1$ in \eqref{eq:lasso_loss} represents the central agent's total payment, and $S_c^{L_1} \left(\mathbf{X}_{\mathcal{N},T}, \boldsymbol{\lambda}_\mathcal{N}^\mathcal{D}, \boldsymbol{\beta} \right)$ represents the central agent's sum of financial losses and payments in the regression market.

\begin{proposition} \label{thrm:payment_guarantee}
If no shrinkage is applied to the central agent's own features ($\lambda_c^d=0, \forall d$),  and each support agent sets $\lambda_i^d = \frac{T}{2} u_i^d$, then the central agent's sum of financial losses and data payments in the regression market is no greater than the central agent's financial losses without the regression market.
\end{proposition}

\begin{proof}
In order to avoid shrinkage of the central agent's own features, we construct scalar matrix $ \boldsymbol{\lambda}_\mathcal{N}^\mathcal{D} =  \text{diag}(0 \ \lambda_1^1 \ \ldots \ \lambda_1^\mathcal{D} \ \lambda_2^1 \ \ldots \ \lambda_2^\mathcal{D} \ \ldots \ \lambda_N^1 \ \ldots \ \lambda_N^\mathcal{D})$, while
\begin{equation}
\lambda_c^d = 0, \ d = 1, 2, \ldots, \mathcal{D} \ .
\label{eq:lambda_cen_0}
\end{equation}

Let us construct a vector of regression parameters $\hat{\boldsymbol{\beta}} \in \mathbb{R}^{1 + N \mathcal{D}} := [\hat{\beta}^0 \ {\hat{\boldsymbol{\beta}}_1}^\top \ \ldots \ {\hat{\boldsymbol{\beta}}_N}^\top]^\top$, where $\hat{\boldsymbol{\beta}}_i = [\hat{\beta}_i^1 \ \hat{\beta}_i^2 \ \ldots \ \hat{\beta}_i^\mathcal{D}]^\top, \forall i \in \mathcal{N}$. We set 
\begin{equation}
\hat{\boldsymbol{\beta}}_i = \boldsymbol{0}^\mathcal{D} = [\underbrace{0 \ \ldots \ 0}_\mathcal{D}]^\top, \forall i \neq c \ .
\label{eq:beta_sup_0}
\end{equation}
We further set
\begin{equation}
[\hat{\beta^0} \ {\hat{\boldsymbol{\beta}}_c}^\top]^\top = \boldsymbol{\beta}_c^\ast \stackrel{\eqref{eq:multi_t_cen_ols}}{=}  \argmin_{\boldsymbol{\beta} \in \mathbb{R}^{1 + \mathcal{D}}} \left\| \mathbf{y}_c - \mathbf{X}_{\{c\},T} \boldsymbol{\beta} \right\|_2^2 \ .
\label{eq:constr_beta_ols_1}
\end{equation}
Therefore, the central agent's average financial losses per time step using their own features are given by 
\begin{equation}
S_c^\text{MSE} \left(\mathbf{X}_{\{c\},T},  \boldsymbol{\beta}_c^\ast \right) =  \frac{1}{T} \left\| \mathbf{y}_c - \left[\mathbf{1}_T \ \mathbf{X}_{c,T}\right] [\hat{\beta^0} \ {\hat{\boldsymbol{\beta}}_c}^\top]^\top \right\|_2^2 \ .
\label{eq:constr_beta_ols_2}
\end{equation}

Applying $\hat{\boldsymbol{\beta}}$ to the lasso loss function of the central agent using all agents' features, we have
\begin{align}
& S_c^{L_1} \left( \mathbf{X}_{\mathcal{N},T}, \boldsymbol{\lambda}_\mathcal{N}^\mathcal{D}, \hat{\boldsymbol{\beta}} \right) \label{eq:lasso_loss_any} \\
\stackrel{\eqref{eq:lasso_loss}}{=} & \frac{1}{T} \left\| \mathbf{y}_c -  \mathbf{X}_{\mathcal{N},T} \hat{\boldsymbol{\beta}} \right\|_2^2 + \left\| \frac{2}{T} \boldsymbol{\lambda}_\mathcal{N}^\mathcal{D} \hat{\boldsymbol{\beta}} \right\|_1 \\
\stackrel{\eqref{eq:beta_sup_0}}{=} & \frac{1}{T} \left\| \mathbf{y}_c - \left[\mathbf{1}_T \ \mathbf{X}_{c,T}\right] [\hat{\beta^0} \ {\hat{\boldsymbol{\beta}}_c}^\top]^\top \right\|_2^2 + \left\| \frac{2}{T} \text{diag}(0 \ \lambda_c^1 \ \ldots \ \lambda_c^\mathcal{D}) [\hat{\beta^0} \ {\hat{\boldsymbol{\beta}}_c}^\top]^\top \right\|_1 \\
\stackrel{\eqref{eq:lambda_cen_0}}{=} & \frac{1}{T} \left\| \mathbf{y}_c - \left[\mathbf{1}_T \ \mathbf{X}_{c,T}\right] [\hat{\beta^0} \ {\hat{\boldsymbol{\beta}}_c}^\top]^\top \right\|_2^2 + 0 \\
\stackrel{\eqref{eq:constr_beta_ols_2}}{=} & S_c^\text{MSE} \left(\mathbf{X}_{\{c\},T},  \boldsymbol{\beta}_c^\ast \right) \label{eq:pr_mse_cen} \\
\stackrel{\eqref{eq:lasso_loss_any}}{\geq} & \min_{\boldsymbol{\beta} \in \mathbb{R}^{1 + N \mathcal{D}}} S_c^{L_1} \left( \mathbf{X}_{\mathcal{N},T}, \boldsymbol{\lambda}_\mathcal{N}^\mathcal{D}, \boldsymbol{\beta} \right) \\
\stackrel{\eqref{eq:reg_market_lasso}}{=} & S_c^{L_1} \left( \mathbf{X}_{\mathcal{N},T}, \boldsymbol{\lambda}_\mathcal{N}^\mathcal{D}, \boldsymbol{\beta}_\mathcal{N}^{L_1}  \right) \label{eq:pf_lasso_loss_all} \\
\stackrel{\eqref{eq:lasso_loss}}{=} & S_c^\text{MSE} \left(\mathbf{X}_{\mathcal{N},T}, \boldsymbol{\beta}_\mathcal{N}^{L_1} \right) +  \left\| \frac{2}{T} \boldsymbol{\lambda}_\mathcal{N}^\mathcal{D} \boldsymbol{\beta}_\mathcal{N}^{L_1} \right\|_1 \label{eq:min_lasso_loss_all} \\
= \ & S_c^\text{MSE} \left(\mathbf{X}_{\mathcal{N},T}, \boldsymbol{\beta}_\mathcal{N}^{L_1} \right) + \sum_{i\in \mathcal{N}_{- c}} \sum_{d=1}^{\mathcal{D}} {\left| u_i^d {\beta_i^d}^{L_1} \right|} ,\label{eq:min_lasso_loss_all_2}
\end{align}
where we obtain \eqref{eq:min_lasso_loss_all_2} from \eqref{eq:min_lasso_loss_all} by setting $\lambda_i^d = \frac{T}{2} u_i^d$.
Since \eqref{eq:min_lasso_loss_all_2} is the central agent's sum of financial losses and data payments in the regression market, and \eqref{eq:pr_mse_cen} is the central agent's financial losses without the regression market, we have thus proved Proposition \ref{thrm:payment_guarantee}.
\end{proof}

Note that with \eqref{eq:pr_mse_cen} $\geq$ \eqref{eq:min_lasso_loss_all_2}, we also prove that \eqref{eq:central_IR} is satisfied, which means the central agent is guaranteed to financially benefit from the regression market.

To further explain this, let us consider a support feature $x_{i, t}^d$ for estimating $y_{c,t}$. In order for the market operator \eqref{eq:reg_market_lasso} not to set the corresponding coefficient $\beta_i^d$ to zero, it has to contribute to a reduction in financial losses that is greater than or equal to the payment it incurs.  Using the assumption $\lambda_i^d = \frac{T}{2} u_i^d$ in Proposition \ref{thrm:payment_guarantee}, the payment to agent $i$ for feature $x_{i, t}^d$ is
\begin{equation}
\left| \frac{2}{T} \lambda_i^d {\beta_i^d}^{L_1} \right| = \left| u_i^d {\beta_i^d}^{L_1} \right| ,
\label{eq:payment_feature}
\end{equation}
meeting the payment requirement of each support agent $i$. In summary, the proposed regression market guarantees financial viability for both the support agents and the central agent. 

Observe that a direct relationship is drawn between the linear coefficients of support features and the payment, without the support features being standardized in the lasso regression. This means that a support feature's mean and variance can influence the linear coefficient,  thus the payment as well. Therefore, each support agent, given a desirable revenue threshold, needs to set their reservation to sell accordingly to account for the means and variances of their support features. The reservation to sell can be considered an independent choice of each support agent with the aim to achieve a certain revenue threshold rather than a metric to compare laterally support agents' willingness to disclose their data. An alternative way of designing the market is to require each support agent to submit standardized data and the corresponding reservation to sell. We choose the former to provide the possibility for support agents to encrypt both their data and their true reservation to sell before submitting features to the market \cite{Obst2017}.

The linear relationship between the payment and the absolute value of the coefficient also raises the question of whether other forms of the payment term (e.g., a quadratic term in $\beta$) could apply to the regression market. Albeit it not being the focus of the paper, the lasso regularizer can also enable the regression with support agents' data to be conducted in a distributed manner \cite{Messner2019}, thus providing another level of privacy protection.  Therefore, we apply the lasso payment in our proposed framework, with the aim to extend it to distributed learning in the future.

\section{Case Studies} \label{sec:cases}

To verify the performance of the proposed regression market, we design two case studies in this section. In the first case study, we construct synthetic datasets for all the agents with fixed linear correlations. This way we have the ground truth of the correlations between the agents' data to verify the model results. In the second case study, we use the hourly wind generation zonal data from Nord Pool to demonstrate the impact of the support agents' reservation to sell on the profit of the market participants.

\subsection{Regression Market Implemented on Synthetic Data}

In this case study, we use an autoregressive process to simulate the data of 4 independent market players (P2-P5) with first order autocorrelations, meaning that their data only have a linear correlation with their own data from the previous time step. Then we use a vector autoregressive process to simulate one market player (P1) that holds data with first order correlations with the other agents. Assuming each time step is one hour, the synthetic data are generated by
\begin{equation*}
\mathbf{z}_{t} = 
\begin{bmatrix}
0.08 & 0.18 & 0.16 & 0.14 & 0.12 \\
0 & 0.9 & 0 & 0 & 0 \\
0 & 0 & 0.8 & 0 & 0 \\
0 & 0 & 0 & 0.7 & 0 \\
0 & 0 & 0 & 0 & 0.6
\end{bmatrix}
\mathbf{z}_{t-1} + \boldsymbol{\varepsilon}_t  ,
\end{equation*}
where $\mathbf{z}_{t} = [z_{1, t} \ z_{2, t} \ \cdots \ z_{5,t}]^\top$ denote both the target and feature values of players (P1-P5), and $\mathbf{\varepsilon}_{t} = [\varepsilon_{1, t} \ \varepsilon_{2, t} \ \cdots \ \varepsilon_{5,t}]^\top$ represent the error terms: $\varepsilon_{i,t} \sim \mathcal{N}(0, 1)$.

\begin{figure} [t]
\begin{center}
\includegraphics[width=10cm]{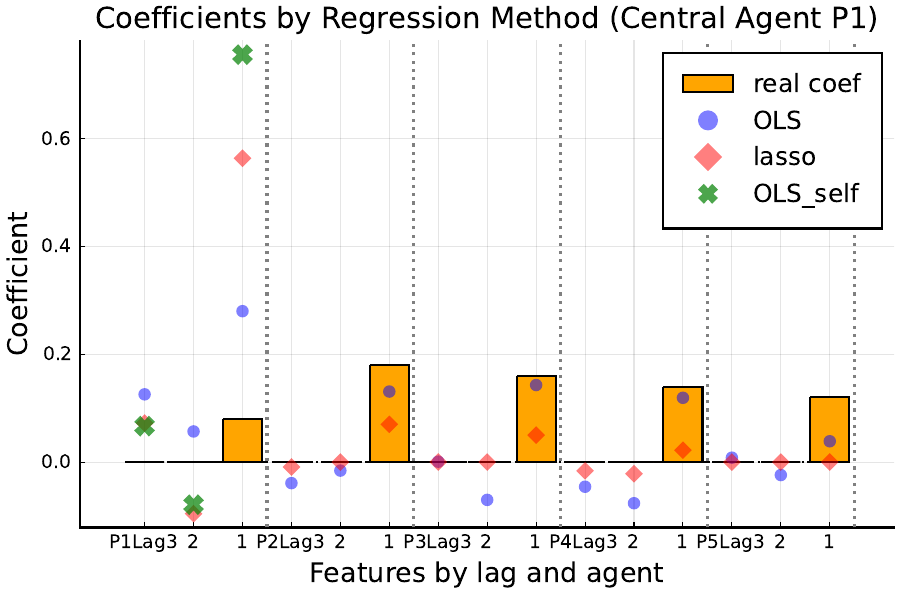}    
\caption{Comparison of regression coefficients by regression method of P1 that has a target variable highly correlated with support agents' data} 
\label{fig:coef_regr_method_cor}
\end{center}
\end{figure}

\begin{figure} [t]
\begin{center}
\includegraphics[width=10cm]{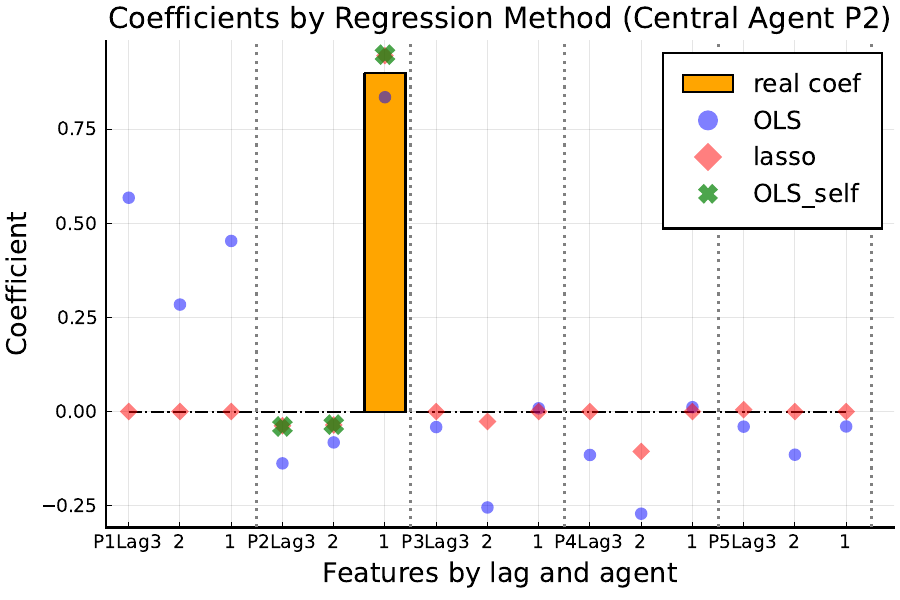}    
\caption{Comparison of regression coefficients by regression method of P2 that has a target variable uncorrelated with support agents' data} 
\label{fig:coef_regr_method_uncorr}
\end{center}
\end{figure}

Focusing on P1 and P2 as central agents, we let the market operator conduct their analytics tasks using different regression methods and obtain the estimated coefficients to compare with the known real coefficients. With a fixed maximum lag of 3 hours and a total training time of 10 days, the market operator returns the results shown in Fig. \ref{fig:coef_regr_method_cor} for P1 and Fig. \ref{fig:coef_regr_method_uncorr} for P2.  Three regression methods are compared: (a) OLS regression on the central agent's own data, (b) OLS regression with all agents' data, and (c) lasso regression with all agents' data. For P1, method (a) cannot capture the correlations with the support agents, and method (b) overfits the data by estimating non-zero coefficients on features that are not correlated with the central agent's data. Method (c) most successfully identifies the non-zero coefficients, albeit shrinking their values due to the lasso regularization.  For P2,  method (b) overfits the data again, while method (c) yields similar results as method (a), well capturing the independence as well as the first-order autocorrelation of P2's data. 

\begin{figure*}[t]
\begin{center}
\includegraphics[width=\textwidth]{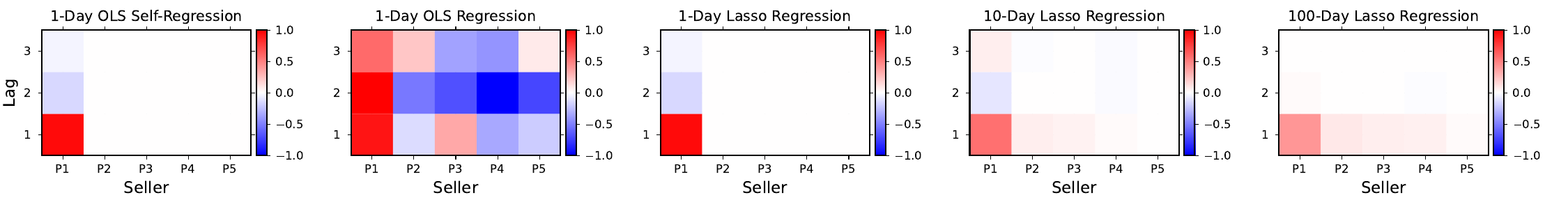}
\caption{Regression coefficients of P1 that has a target variable correlated with support agents' data, varying by the time scale of regression.}
\label{fig:beta_hmap_corr}
\end{center}
\end{figure*}

\begin{figure*}[t]
\begin{center}
\includegraphics[width=\textwidth]{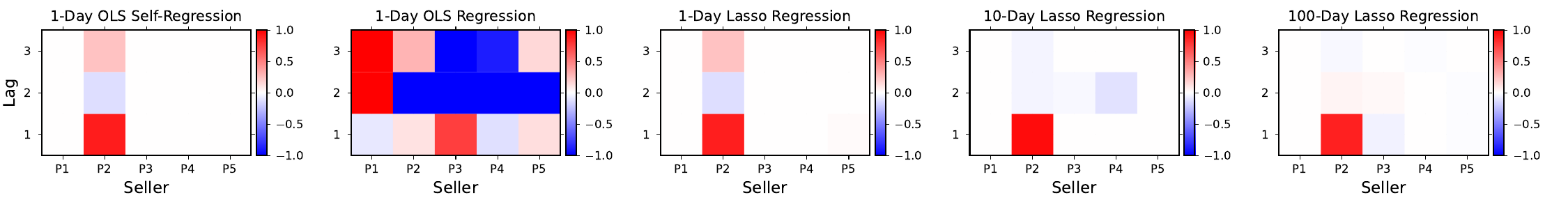}
\caption{Regression coefficients of P2 that has a target variable uncorrelated with support agents' data, varying by the time scale of regression.}
\label{fig:beta_hmap_uncorr}
\end{center}
\end{figure*}

To see how the results from the the regression change as the training time increases, Fig.  \ref{fig:beta_hmap_corr} and Fig.  \ref{fig:beta_hmap_uncorr} provide a visual representation of the estimated coefficients.  For P1, lasso regression demonstrates a clear advantage over OLS regression for reducing overfitting, and over OLS self-regression for being able to capture correlations with features from support agents. For P2, even though the lasso regression cannot provide much additional benefit to the central agent, it is still effective in identifying the agent's independence and reducing overfitting. 

\begin{figure} [t]
\begin{center}
\includegraphics[width=14cm]{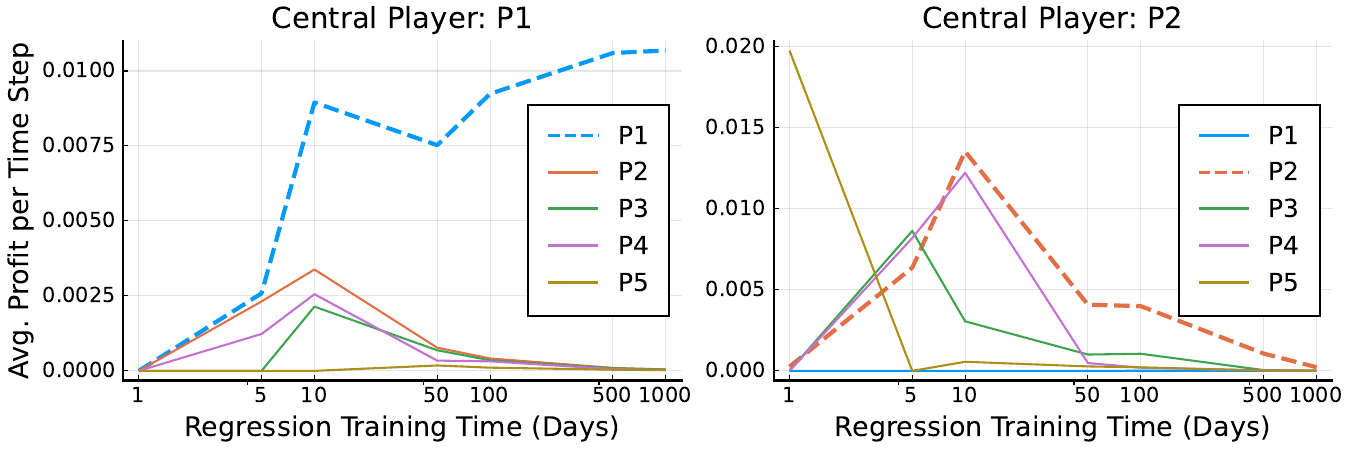}    
\caption{Average profit per time step of all players with P1 and P2 as the central agents, respectively.} 
\label{fig:prof_corr}
\end{center}
\end{figure}

Next, we evaluate the impact of the training time on the profits of the agents. Here, we reiterate that the analyses done in this paper are in-sample, meaning the central agent's profit only reflects the forecast improvement during the training period. As Fig. \ref{fig:prof_corr} shows, increasing the training time $T$ eventually reduces the support agents' average profit per time step to zero regardless of the correlations between the support agents and the central agent. This is because the payment term $\left| u_i^d {\beta_i^d}^{L_1} \right|$ from \eqref{eq:payment_feature} does not scale with $T$. In practice, in order to fulfill a certain revenue threshold per time step, the support agents can take $T$ into account when setting their reservation to sell (e.g., setting the $u$ value in proportion to the length of their offered support features).  Meanwhile, P1 and P2 as central agents have very different profit trajectories as $T$ increases. For P1, as their data are highly correlated with the support agents, more training time improves the estimates of the coefficients, hence increasing their profit. For P2, as their data are independent of others, the improvement of their analytics task is limited to having the lasso shrink the untrue coefficients of their own data, and as the training time increases, this minor improvement also reduces to zero. Note that in reality, an agent with data that are completely independent from others would not have the motivation to participate in a data market.

\subsection{Regression Market Implemented on Real Data}

To test our proposed regression market on real-world data, we obtain zonal wind power data for Denmark and Sweden from the open source Nord Pool data repository. As an illustrative example, these data are aggregated in six zones\footnote{https://www.nordpoolgroup.com/the-power-market/Bidding-areas/} (DK1, DK2, SE1, SE2, SE3, and SE4) to represent six players in the regression market. In practice, since the clearing of the regression market is decoupled from the energy market, a central agent with assets in multiple energy market zones can participate in the regression market by adjusting their loss function to truthfully reflect their financial situation in the energy market. The main purpose of this case study is to examine the impact of the support agents' reservation to sell on the final market outcomes. 

\begin{figure} [t]
\begin{center}
\includegraphics[width=10cm]{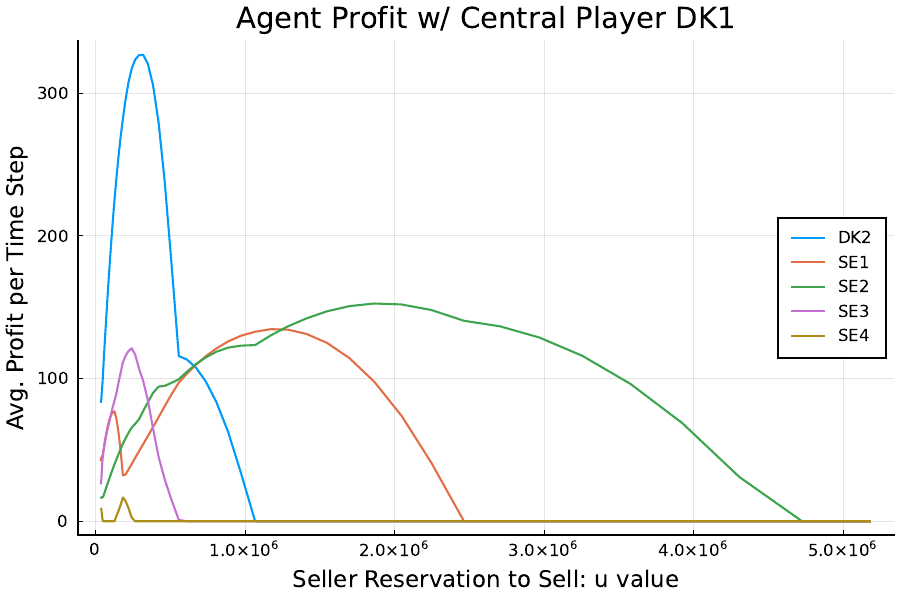}
\caption{Support agent profits with varying $u$ values (central agent: DK1, training time: 10 days).}
\label{fig:prof_lambda}
\end{center}
\end{figure}

Assigning DK1 as the central agent, and the others as support agents for the regression market, we use $1$ hour as the time step, $1$ hour ahead as the forecast horizon, $5$ hours as the maximum lag, and $10$ days ($240$ hours) as the training time. First, we assume all the support agents' reservation to sell ($u_i^d$) for all the features to be the same value $u$. Varying its value, we plot the payments for all the support agents in Fig. \ref{fig:prof_lambda}. In general, every agents' profit first increases with the $u$ value, then peaks, and then reduces to zero. To explain this, we recall the payment term $\left| u_i^d {\beta_i^d}^{L_1} \right|$. When $u_i^d = u = 0$, the payment is zero, but $\beta_i^d$ is at its peak due to a lack of shrinkage. As $u$ increases, the profit increases, but meanwhile more shrinkage is applied to $\beta_i^d$. A peak appears when the trade-off between the two opposite forces reaches a balance, but afterwards the lasso gradually shrinks $\beta_i^d$ to zero, and the profit becomes zero again.  The position of the peak may have to do with the correlation and the magnitude of the support agent's data. This is an interesting topic for future work.

\begin{figure} [t]
\begin{center}
\includegraphics[width=14cm]{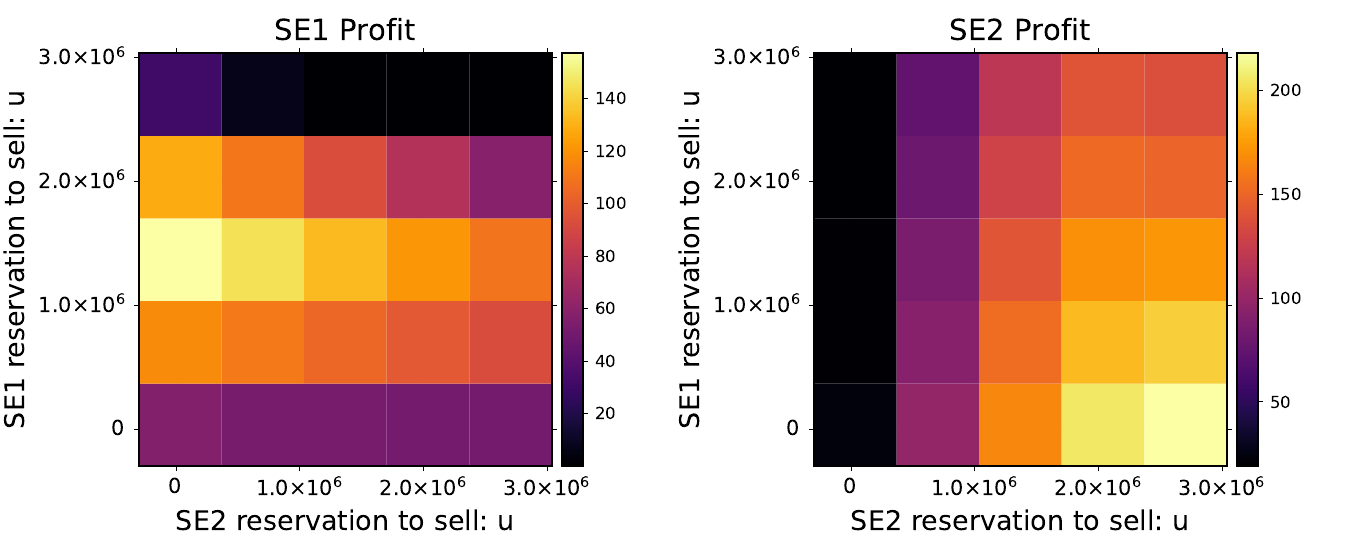}
\caption{SE1 and SE2 profits with varying $u$ values of SE1 and SE2 (central agent: DK1, training time: 10 days).}
\label{fig:prof_lambda_competition}
\end{center}
\end{figure}

Lastly, we select from Fig. \ref{fig:prof_lambda} the two agents whose profits peak the last to analyze the mutual impact of their reservation to sell on each other's profit. Here, we allow the reservation to sell for SE1 and SE2 to be different while fixing the other agents' $u$ value, and show their resulting profits in Fig.  \ref{fig:prof_lambda_competition}. It is observed that within the demonstrated range, SE2 benefits from higher $u$ values, while SE1 achieves the highest profit in the middle range. Meanwhile, each player's profit decreases as the other player's reservation to sell increases. So far, the reservation to sell has been used as a customized parameter by the support agents to ensure a revenue threshold, but the results from Fig. \ref{fig:prof_lambda_competition} raise the question of whether the support agents could instead strategically set the reservation to sell to maximize their gain in the regression market. Future research can examine the interplay of the reservation to sell of more than two players and the corresponding market equilibrium.

\section{Conclusion} \label{sec:conclusion}

Adopting the lasso regularizer, we construct a regression market for wind agents to trade wind power data to improve forecasting.  Each support agent as a data seller has the freedom to determine their reservation to sell each feature they own, which is incorporated in the lasso term of the central agent's analytics task under linear regression. The product of a support agent's reservation to sell a feature and the absolute value of the corresponding estimated coefficient is directly computed as the payment from the central agent for the sold feature.  This market framework is proved to meet the support agents' profit requirements while guaranteeing financial benefits for the central agent.  

Some immediate future work includes the incorporation of out-of-sample analyses, the strategies to set the reservation to sell from the data sellers' perspective for maximizing individual gains, and the extension to an online market. This regression market can be applied to other use cases, where the data sellers have individual revenue requirements on the data sold to the data buyer. It can also be readily extended to a multi-buyer framework since the analytics tasks of multiple agents can be conducted simultaneously and the outcome of each agent's task does not affect the task of another.

\section*{Acknowledgements} 

The research leading to this work is being carried out as a part of the Smart4RES project (European Union’s Horizon 2020, No. 864337). The
sole responsibility of this publication lies with the author. The European
Union is not responsible for any use that may be made of the information
contained therein.

The open source and easily accessible market data from Nord Pool have provided the basis for case studies in this work. 

In addition, the authors would like to acknowledge the fruitful discussions on data markets with Ricardo J. Bessa and Carla Gonçalves at INESC Porto, as well as the useful tips on autoregressive models provided by Amandine Pierrot at DTU.

\bibliographystyle{myIEEEtran}
\bibliography{PSCC_2022.bib}

\begin{thebibliography}{10}
\providecommand{\url}[1]{#1}
\csname url@rmstyle\endcsname
\providecommand{\newblock}{\relax}
\providecommand{\bibinfo}[2]{#2}
\providecommand\BIBentrySTDinterwordspacing{\spaceskip=0pt\relax}
\providecommand\BIBentryALTinterwordstretchfactor{4}
\providecommand\BIBentryALTinterwordspacing{\spaceskip=\fontdimen2\font plus
\BIBentryALTinterwordstretchfactor\fontdimen3\font minus
  \fontdimen4\font\relax}
\providecommand\BIBforeignlanguage[2]{{%
\expandafter\ifx\csname l@#1\endcsname\relax
\typeout{** WARNING: IEEEtran.bst: No hyphenation pattern has been}%
\typeout{** loaded for the language `#1'. Using the pattern for}%
\typeout{** the default language instead.}%
\else
\language=\csname l@#1\endcsname
\fi
#2}}

\bibitem{Veldkamp2020}
L.~Veldkamp, ``Data and the aggregate economy,'' Society for Economic
  Dynamics,'' Annual Meeting Plenary, 2019.

\bibitem{Morales2014}
J.~M. Morales, A.~J. Conejo, H.~Madsen, P.~Pinson, and M.~Zugno,
  \emph{\BIBforeignlanguage{English}{Integrating Renewables in Electricity
  Markets: Operational Problems}}.\hskip 1em plus 0.5em minus 0.4em\relax
  Springer, 2014.

\bibitem{Jung2014}
J.~Jung and R.~P. Broadwater, ``{Current status and future advances for wind
  speed and power forecasting},'' \emph{Renewable and Sustainable Energy
  Reviews}, vol.~31, pp. 762--777, 2014.

\bibitem{Tastu2010}
J.~Tastu, P.~Pinson, and H.~Madsen, ``\BIBforeignlanguage{English}{Multivariate
  conditional parametric models for a spatio-temporal analysis of short-term
  wind power forecast errors},'' in
  \emph{\BIBforeignlanguage{English}{Proceedings of the European Wind Energy
  Conference}}, 2010.

\bibitem{Andrade2017}
J.~R. Andrade and R.~J. Bessa, ``Improving renewable energy forecasting with a
  grid of numerical weather predictions,'' \emph{IEEE Transactions on
  Sustainable Energy}, vol.~8, no.~4, pp. 1571--1580, 2017.

\bibitem{Farokhi2020}
F.~Farokhi, ``Review of results on smart-meter privacy by data manipulation,
  demand shaping, and load scheduling,'' \emph{IET Smart Grid}, vol.~3, no.~5,
  pp. 605--613, 2020.

\bibitem{Goncalves2020}
C.~Goncalves, P.~Pinson, and R.~J. Bessa, ``{Towards data markets in renewable
  energy forecasting},'' \emph{IEEE Transactions on Sustainable Energy},
  vol.~12, no.~1, pp. 533--542, 2021.

\bibitem{Bergemann2019}
D.~Bergemann and A.~Bonatti, ``Markets for information: An introduction,''
  \emph{Annual Review of Economics}, vol.~11, pp. 85--107, 2019.

\bibitem{Pinson2022}
\BIBentryALTinterwordspacing
P.~Pinson, L.~Han, and J.~Kazempour, ``Regression markets and application to
  energy forecasting,'' \emph{in press, TOP}, 2022, preprint available.
  [Online]. Available: \url{https://arxiv.org/abs/2110.03633}
\BIBentrySTDinterwordspacing

\bibitem{Montes2019}
R.~Montes, W.~Sand-Zantman, and T.~Valletti, ``{The value of personal
  information in online markets with endogenous privacy},'' \emph{Management
  Science}, vol.~65, no.~3, pp. 1342--1362, 2019.

\bibitem{Acemoglu2021}
D.~Acemoglu, A.~Makhdoumi, A.~Malekian, and A.~Ozdaglar, ``Too much data:
  Prices and inefficiencies in data markets,'' \emph{forthcoming, American
  Economic Journal: Microeconomics:Micro}, 2021.

\bibitem{Bimpikis2019}
K.~Bimpikis, D.~Crapis, and A.~Tahbaz-Salehi, ``Information sale and
  competition,'' \emph{Management Science}, vol.~65, no.~6, pp. 2646--2664,
  2019.

\bibitem{Agarwal2019}
A.~Agarwal, M.~Dahleh, and T.~Sarkar, ``A marketplace for data: An algorithmic
  solution,'' \emph{Proceedings of the 2019 ACM Conference on Economics and
  Computation}, pp. 701--726, 2019.

\bibitem{Girard2013}
R.~Girard and D.~Allard, ``Spatio‐temporal propagation of wind power
  prediction errors,'' \emph{Wind Energy}, vol.~16, no.~7, pp. 999--1012, 2013.

\bibitem{He2014}
M.~He, L.~Yang, J.~Zhang, and V.~Vittal, ``{A spatio-temporal analysis approach
  for short-term forecast of wind farm generation},'' \emph{IEEE Transactions
  on Power Systems}, vol.~29, no.~4, pp. 1611--1622, 2014.

\bibitem{Pinson2016}
P.~Pinson, ``{Introducing distributed learning approaches in wind power
  forecasting},'' \emph{2016 International Conference on Probabilistic Methods
  Applied to Power Systems}, 2016.

\bibitem{Boyd2010}
S.~Boyd, N.~Parikh, E.~Chu, B.~Peleato, and J.~Eckstein, ``{Distributed
  optimization and statistical learning via the alternating direction method of
  multipliers},'' \emph{Foundations and Trends in Machine Learning}, vol.~3,
  no.~1, pp. 1--122, 2010.

\bibitem{Tibshirani1996}
R.~Tibshirani, ``Regression shrinkage and selection via the lasso,''
  \emph{Journal of the Royal Statistical Society. Series B (Methodological)},
  vol.~58, no.~1, pp. 267--288, 1996.

\bibitem{Cavalcante2017}
L.~Cavalcante, R.~J. Bessa, M.~Reis, and J.~Browell, ``Lasso vector
  autoregression structures for very short-term wind power forecasting,''
  \emph{Wind Energy}, vol.~20, no.~4, pp. 657--675, 2017.

\bibitem{Messner2019}
J.~W. Messner and P.~Pinson, ``{Online adaptive lasso estimation in vector
  autoregressive models for high dimensional wind power forecasting},''
  \emph{International Journal of Forecasting}, vol.~35, no.~4, pp. 1485--1498,
  2019.

\bibitem{Obst2017}
\BIBentryALTinterwordspacing
D.~Obst and P.~Pinson, ``{Distributed learning for high-dimensional wind energy
  production forecasting},'' Technical University of Denmark, Tech. Rep. 1--61,
  2017. [Online]. Available:
  \url{https://davidobst.github.io/Rapport\_PRe\_Obst\_David.pdf}
\BIBentrySTDinterwordspacing

\end{thebibliography}

\end{document}